\documentclass[sigconf,nonacm]{acmart}

\usepackage{pifont}
\usepackage{algorithm}
\usepackage{algpseudocode}
\usepackage{booktabs}
\usepackage{multirow}
\usepackage{amsmath}
\usepackage{graphicx}
\usepackage{xcolor}
\usepackage{enumitem}
\usepackage{microtype}
\usepackage{url}
\usepackage{subcaption}
\newcommand{\benchYes}{\textcolor{green!55!black}{\ding{51}}}
\newcommand{\benchNo}{\textcolor{red!75!black}{\ding{55}}}

\newcommand{\abYes}{\ding{51}}
\newcommand{\abNo}{\ding{55}}
\setcopyright{none}
\begin{document}

\title{BrowseFlow: Sparse Event Flow Matching for Generative User Simulation}

\author{Yitian Pang}
\authornote{These authors contributed equally to this work.}
\email{pangyt23@mails.tsinghua.edu.cn}
\affiliation{%
    \department{Department of Automation}
    \institution{Tsinghua University}
    \city{Beijing}
    \country{China}
}

\author{Jiayi Zhang}
\authornotemark[1]
\email{luohan.zjy@alibaba-inc.com}
\affiliation{%
    \institution{Taobao \& Tmall Group of Alibaba}
    \city{Beijing}
    \country{China}
}

\author{Zhenjia Ma}
\email{mazhenjia.mzj@alibaba-inc.com}
\affiliation{%
    \institution{Taobao \& Tmall Group of Alibaba}
    \city{Beijing}
    \country{China}
}

\author{Xiaowei Yue}
\authornote{Corresponding authors.}
\email{yuex@tsinghua.edu.cn}
\affiliation{%
    \department{Department of Industrial Engineering}
    \institution{Tsinghua University}
    \city{Beijing}
    \country{China}
}

\author{Dagui Chen}
\authornotemark[2]
\email{dagui.cdg@alibaba-inc.com}
\affiliation{%
    \institution{Taobao \& Tmall Group of Alibaba}
    \city{Beijing}
    \country{China}
}

\author{Jian Xu}
\email{xiyu.xj@alibaba-inc.com}
\affiliation{%
    \institution{Taobao \& Tmall Group of Alibaba}
    \city{Beijing}
    \country{China}
}

\renewcommand{\shortauthors}{Pang et al.}

\begin{abstract}
User simulation aims to reproduce how users interact with content and
products across online platforms. Recent generative approaches have
enriched user simulation with behavioral history and contextual
information. However, realistic browsing simulation still faces
two challenges.
(i) \textbf{Stochastic Responses}: 
Even with the same recommendation feed, the same user
may produce different browsing responses, yielding
distinct behavioral trajectories. A logged trajectory records only one such
realization. Reference matching alone cannot establish whether other
plausible browsing responses are assigned appropriate probabilities.
(ii) \textbf{Sparse Key Events}: Clicks and likes provide important
feedback but occur sparsely within browsing trajectories dominated
by progression actions such as scrolling. Many stepwise simulators
predict these actions in a dense sequence, where key events contribute
relatively few training targets.
To address these challenges, we propose \textbf{BrowseFlow}, a conditional
discrete flow matching method for generative user simulation.
For Sparse Key Events, BrowseFlow represents each response
by its browsing extent and an ordered sequence of key events,
directly modeling their counts and locations under exposure constraints.
BrowseFlow learns the conditional distribution of complete browsing
responses through probabilistic extent prediction and discrete flow
matching over key events. Random event retention defines
probability paths over partial event sequences, while transition rate
matching learns stochastic insertion dynamics.
For Stochastic Responses, we use the energy score, a proper
scoring rule, to evaluate response fidelity and generated
variability, complemented by the Wasserstein distance between
key event count distributions.
We further introduce \textbf{TaobaoBrowse}, a benchmark pairing real
recommendation contexts with fine-grained browsing trajectories.
Experiments across three datasets demonstrate that BrowseFlow
consistently achieves the best distributional performance.
\end{abstract}



\maketitle

\begingroup
\raggedbottom
\setlength{\textfloatsep}{7pt plus 2pt minus 2pt}
\setlength{\dbltextfloatsep}{7pt plus 2pt minus 2pt}


\section{Introduction}
\label{sec:intro}

User simulation, with the aim of reproducing how users interact with
content and products across online platforms, has been widely explored
in e-commerce, short-video, news, and other interactive services.
Users engage in diverse activities such as scrolling through feeds,
clicking on items, and liking content, forming behavioral trajectories
through successive interactions. Modeling such trajectories is a
longstanding research problem in understanding and simulating user
behavior~\cite{balog2023user}.
Early approaches primarily focused on modeling user feedback over
ranked items. Classical click models characterize examination and
click decisions~\cite{dupret2008user,chapelle2009dynamic}.
Simulation frameworks such as RecSim and RecSim NG broadened
the scope to sequential user interactions, providing environments
for modeling user behavior over time
~\cite{ie2019recsim,mladenov2021recsim}.

More recently, generative approaches have used stepwise
generation to model user behavior with richer contextual detail.
Data-driven simulators such as Virtual-Taobao, UserSim, and
KuaiSim learn patterns of browsing, feedback, and session
continuation from interaction logs
~\cite{shi2019virtual,zhao2021usersim,zhao2023kuaisim}.
Large language models incorporate user profiles, memory, and
preference reasoning into behavioral
simulation~\cite{wang2025user,zhang2024generative,zhang2025llm}.
Vision-language models ground simulated interactions in
webpage observations~\cite{zhang2025see}, while
RecVerse~\cite{tang2026towards} combines hierarchical memory
with trajectory-level optimization for sequential shopping
simulation. These developments have enabled richer
representations of browsing responses across diverse user
and recommendation contexts. Despite these advances,
realistic browsing simulation still faces two challenges.

\begin{figure}[t]
    \centering
    \includegraphics[width=\columnwidth]{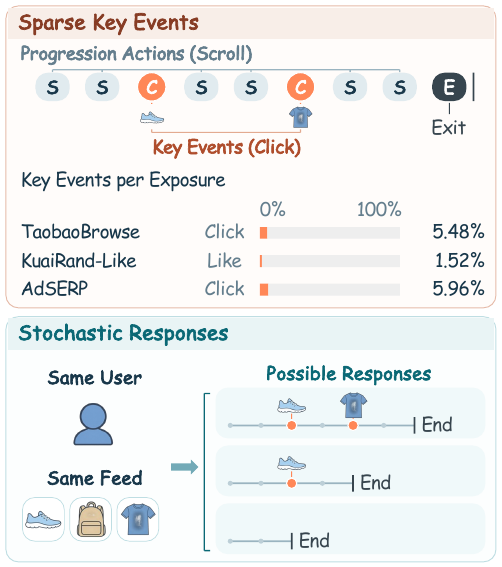}

    \caption{\textbf{Motivation of BrowseFlow.}
    The upper panel distinguishes frequent progression
    actions from sparse key events, highlighting the
    imbalance between browsing progression and meaningful
    feedback across three datasets. The lower panel illustrates stochastic responses,
    where the same user may produce different browsing
    trajectories under an identical recommendation feed.}

    \Description{
    The upper panel illustrates a browsing action sequence
    dominated by scrolling, with sparse clicks on items.
    A bar chart reports the proportions of key events
    across TaobaoBrowse, KuaiRand-Like, and AdSERP,
    showing positive feedback rates of 5.48 percent,
    1.52 percent, and 5.96 percent, respectively.
    The lower panel shows the same user and recommendation
    feed leading to three possible browsing responses with
    different clicked items and termination depths.
    }

    \label{fig:motivation}
\end{figure}

\textbf{Stochastic Responses.}
User browsing behavior is inherently stochastic.
As illustrated in Fig.~\ref{fig:motivation}, the same
user presented with the same recommendation feed may
produce different behavioral trajectories.
This variability is consistent with random utility
theory~\cite{mcfadden1972conditional} and recent choice
modeling for recommendation~\cite{krause2025non}.
Each logged trajectory captures only one
realization of the conditional response distribution.
Therefore, faithful user simulation aims to model this
distribution rather than reproduce individual logged
trajectories. At the evaluation level, existing methods primarily measure agreement between generated responses and the logged trajectory. Such reference matching metrics emphasize reproducing a single observed realization, providing limited insight into whether the model captures other plausible browsing responses.

\textbf{Sparse Key Events.}
Browsing trajectories exhibit a highly imbalanced action
distribution. Progression actions, such as scrolling and
swiping, advance content exposure and constitute most
browsing actions. In contrast, key events such as clicks
and likes reflect user engagement with specific items,
providing valuable signals for downstream preference
modeling and recommendation optimization.
However, these events occur sparsely: positive feedback
accounts for only 1.52\%--5.96\% of valid exposures
across our three datasets (Fig.~\ref{fig:motivation}).
Prior studies of browsing continuation~\cite{ji2022intra}
and probabilistic scrolling budgets~\cite{jeunen2023probabilistic}
also distinguish exposure progression from engagement.
Many stepwise simulators~\cite{shi2019virtual,zhao2021usersim,zhao2023kuaisim} model both types of
actions within the same dense sequence, where frequent
progression actions dominate prediction targets.
Consequently, key event occurrence and placement receive
comparatively limited direct supervision.
During generation, errors in progression decisions may
alter subsequent exposure states, while premature
termination can prevent later key events from occurring.

To address these challenges, we propose \textbf{BrowseFlow}, a conditional
discrete flow matching method for generative user simulation.
\textbf{For Sparse Key Events}, BrowseFlow represents each response by its
browsing extent and an ordered sequence of key events.
The extent compactly summarizes frequent progression actions, while
the event sequence makes sparse key events the focus of generation.
BrowseFlow models the conditional distribution of complete browsing
responses through probabilistic extent prediction and discrete flow
matching over key event sequences. Random event retention
defines probability paths over partial sequences, while transition
rate matching learns stochastic insertion dynamics.
Conditioned on the predicted extent, the insertion process directly
generates variable event counts and locations under exposure constraints.
\textbf{For Stochastic Responses}, we use the energy
score~\cite{gneiting2007strictly}, a proper scoring rule for predictive
distributions, to evaluate both fidelity to observed responses and
variation across repeated generations. The Wasserstein distance
between key event count distributions complements this evaluation
by measuring aggregate event count alignment.
We further introduce \textbf{TaobaoBrowse}, a benchmark pairing real
recommendation contexts with fine-grained trajectories to support
training and evaluation of browsing responses.

To summarize, our contributions are threefold:
\begin{itemize}
    \item \textbf{Discrete Flow Matching for User Simulation.}
    We introduce conditional discrete flow matching for generative
    user simulation, representing complete browsing responses through
    browsing extent and sparse key events. This formulation captures
    stochastic variation in browsing responses while directly modeling
    key event counts and locations under exposure constraints.

    \item \textbf{Distributional Evaluation of Simulated Responses.}
    We apply the energy score to assess how well simulated
    conditional response distributions agree with observed
    user behavior, considering both observation fidelity
    and generated variability.
    The Wasserstein distance between key event count distributions
    complements this evaluation by measuring aggregate
    event count alignment.

    \item \textbf{The TaobaoBrowse Benchmark.}
    We introduce TaobaoBrowse, a real-world benchmark that preserves
    recommendation sessions, fine-grained browsing behavior, and
    available user and item context for user browsing simulation.
\end{itemize}



\section{Related Work}
\label{sec:related}

\subsection{Generative Recommendation}

Generative recommendation extends conventional candidate scoring
to conditional generation of recommendation outputs.
Autoregressive methods generate item identifiers or semantic
IDs~\cite{petrov2023generative,rajput2023recommender},
while diffusion methods model item
representations~\cite{li2023diffurec,yang2023generate}.
Flow matching has been applied to collaborative filtering and
sequential recommendation~\cite{liu2025flow,liu2025flow_fmrec,li2026flowrec},
and OneRec extends generation to complete recommendation
sessions~\cite{deng2025onerec}.
These approaches focus on generating the recommendations
presented to users. Generative user simulation studies a
complementary task: generating user behavior conditioned on
the presented recommendation context.
BrowseFlow focuses on complete browsing responses,
modeling browsing progression and key events as a
structured conditional generation problem under
exposure constraints.

\subsection{User Simulation and Evaluation}

User simulation has evolved from local feedback prediction
toward generating complete behavioral trajectories.
Classical click models and neural predictors capture examination,
click dependencies, and evolving user
interests~\cite{dupret2008user,chapelle2009dynamic,borisov2016neural,zhou2019deep}.
Simulation platforms support sequential
interactions~\cite{ie2019recsim,mladenov2021recsim},
while learned simulators model browsing behavior and
multiple feedback signals~\cite{zhao2021usersim,zhao2023kuaisim}.
Agent approaches incorporate profiles, memory, and Web
observations~\cite{zhang2024generative,zhang2025see}.
RecVerse introduces trajectory optimization for shopping
simulation~\cite{tang2026towards}, while LPDO models dependencies
among future interactions through listwise
diffusion~\cite{huang2026listwise}.
Evaluating stochastic responses also requires distributional
measures. Proper scoring rules assess predictive
distributions~\cite{gneiting2007strictly,alexander2024evaluating},
including complete trajectories~\cite{golestaneh2016generation}.
Reference matching and best-of-$N$ metrics emphasize individual
matches rather than the full generated
distribution~\cite{shahroudi2024evaluation}.
We use the energy score to evaluate conditional browsing
response distributions, accounting for both response fidelity
and variation across generated trajectories.

\subsection{Flow Matching for Structured Sequences}

Insertion and editing methods enable flexible sequence
generation~\cite{stern2019insertion,gu2019levenshtein}.
Temporal point process models capture event dependencies
through history-dependent
intensities~\cite{du2016recurrent,mei2017neural},
while Add and Thin introduces diffusion-based event
generation~\cite{ludke2023add}.
Flow matching learns generative dynamics along probability
paths~\cite{lipman2022flow}.
Extensions to discrete state spaces formulate these dynamics
through continuous-time Markov
processes~\cite{campbell2024generative,gat2024discrete},
with related work on discrete denoising and distribution
ratios~\cite{campbell2022continuous,lou2023discrete}.
General discrete paths support flexible probability path
constructions~\cite{shaul2025flow}, while Edit Flows extends
discrete flow matching to variable-length sequences through
learned edit operations~\cite{havasi2026edit}.
Building on discrete flow matching, BrowseFlow defines
conditional insertion dynamics over key event sequences.
Its generation process models stochastic event counts and
locations within the predicted browsing extent while
enforcing exposure constraints.


\section{Method}
\label{sec:method}

\subsection{Problem Formulation}
\label{sec:problem}

Generative user simulation aims to generate realistic sequences of user
actions conditioned on observed user information and interaction context.
We study this task for a fixed recommendation session
$S=(v_1,\ldots,v_N)$. The observed context is $X=(h,u,S,f)$, where $h$
denotes the interaction history, $u$ contains available user features,
and $f$ is the current browsing frontier. The complete recommendation
session is known before generation, while subsequent user behavior
remains unknown.

A browsing response $\tau=(a_1,\ldots,a_T)$ consists of three behavioral
components. Progression actions, such as scrolling, advance content
exposure. Key events, such as clicks and likes, reflect user engagement
with specific items. Session termination marks the end of browsing.
The objective is to generate a complete behavioral response conditioned
on the observed context:
\begin{equation}
    \hat{\tau}=(\hat{a}_1,\ldots,\hat{a}_{\hat{T}})
    \sim p_\Theta(\cdot\mid X),
    \label{eq:simulation_objective}
\end{equation}
where $p_\Theta$ denotes the learned generative model. Both the action
sequence and its length are determined through generation, allowing
different browsing responses under the same context.

\subsection{Sparse Response Representation}
\label{sec:sparse_representation}

To focus generation on sparse key events, BrowseFlow represents browsing
progression compactly while modeling key events explicitly. Let
$D\in\{f,\ldots,N\}$ denote the final browsing extent. The ordered key
event sequence is $\mathcal E=(e_1,\ldots,e_M)$, where
$e_k=(d_k,\mu_k)$ records the browsing depth and mark of event $k$.
We represent a complete response by
\begin{equation}
    \mathcal R_X(\tau)=(D,\mathcal E).
    \label{eq:structured_response}
\end{equation}
The event sequence may be empty. Otherwise,
$f\leq d_1\leq\cdots\leq d_M\leq D$, and sequence order distinguishes
events at the same depth. In the click setting, $\mu_k$ identifies an
interacted item. A target must be available by depth $d_k$ and cannot
repeat within the same session.

The representation separates browsing progression from key event
generation. Given $D$ and $\mathcal E$, a deterministic decoder
reconstructs the progression actions between event depths and terminates
browsing at $D$. Thus, progression actions are represented through the
response structure, while sparse key events become the primary
generative variables.

This representation leads to an extent-first factorization:
\begin{equation}
    p_\Theta(D,\mathcal E\mid X)
    =\pi_\eta(D\mid X)\,p_\theta(\mathcal E\mid D,X_D),
    \label{eq:browseflow_factorization}
\end{equation}
where $X_D=(h,u,S_{\leq D},f)$ denotes the context restricted to the
selected browsing extent. The first factor models the distribution
of browsing extent, while the second models the conditional
distribution of key events within that extent. Together, they define
a stochastic generative model over complete browsing responses.
Figure~\ref{fig:pipeline} summarizes the overall framework.

\begin{figure*}[t]
    \centering
    \IfFileExists{fig/pipeline.pdf}{%
        \includegraphics[width=\textwidth]{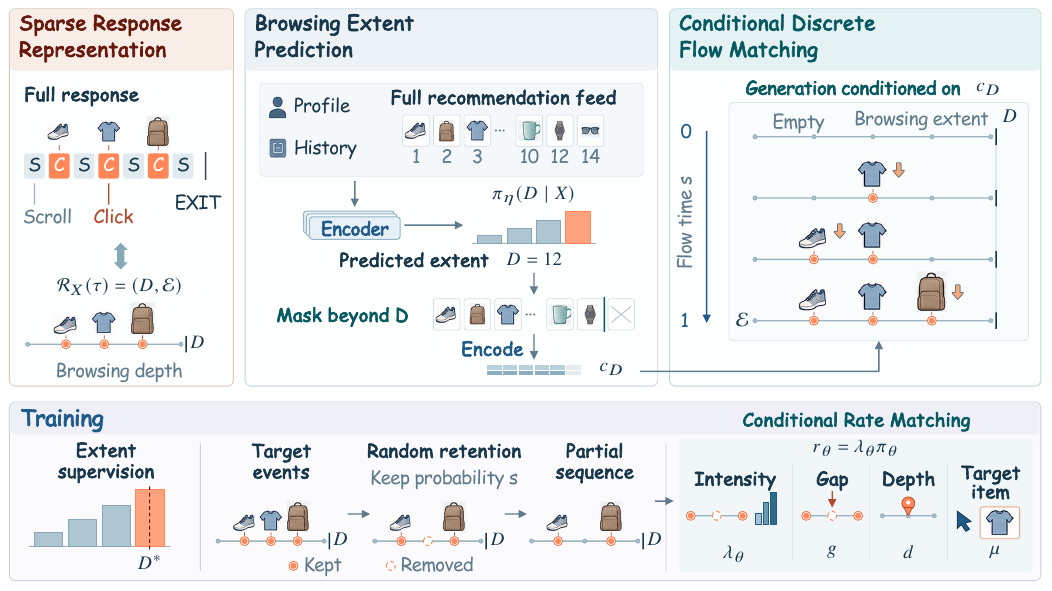}%
    }{%
        \fbox{\parbox[c][1.12in][c]{0.94\textwidth}{\centering
            Complete session + user context $\rightarrow$ extent prediction
            $\rightarrow$ extent-restricted encoding\\[5pt]
            Empty event sequence $\rightarrow$ conditional discrete Flow Matching
            $\rightarrow$ ordered key events $\rightarrow$ response reconstruction\\[5pt]
            Training: observed extent + random event retention
            $\rightarrow$ target insertion rates}}
    }

    \caption{Overview of BrowseFlow.
    The upper panels illustrate the sparse response
    representation, browsing extent prediction, and
    conditional discrete Flow Matching.
    Browsing progression is represented by its extent,
    while key events are generated through stochastic
    insertions under exposure constraints.
    The resulting extent and event sequence determine
    a complete browsing response.
    The lower panel summarizes training, where observed
    extents supervise extent prediction and random event
    retention provides insertion targets for conditional
    rate matching.}

    \Description{A two-stage generative framework. The recommendation
    session and user context are used to predict browsing extent.
    The selected extent conditions a discrete insertion flow
    that generates an ordered sequence of key events. A deterministic
    decoder reconstructs the complete browsing response.
    Training uses observed browsing extents and randomly retained
    event subsequences.}

    \label{fig:pipeline}
\end{figure*}

\subsection{Browsing Extent Prediction}
\label{sec:extent}

BrowseFlow first models the distribution of final browsing extent
conditioned on the complete recommendation session. A session Transformer
encodes all item attributes with position information, while a history
Transformer encodes the observed user interactions. Their representations
are combined with user features to score each reachable extent.
The resulting categorical distribution is
\begin{equation}
    \pi_\eta(D=d\mid X)
    =\frac{\exp\alpha_\eta(d;X)}
    {\sum_{r=f}^{N}\exp\alpha_\eta(r;X)},
    \qquad d\in\{f,\ldots,N\}.
    \label{eq:extent_prediction}
\end{equation}
Here $\alpha_\eta(d;X)$ is the learned extent score. The full
recommendation session is visible to this branch, allowing different
browsing extents to be sampled under the same context.

After sampling $D$, BrowseFlow performs a second encoding pass for key
event generation. Item attributes beyond $D$ are masked before
contextual encoding, while list positions remain available.
The extent and event stages share encoder parameters but do not reuse
contextual activations. We denote the resulting item representations
by $P_j^{(D)}$ and the cached event context, including these
representations, by $c_D$. Both remain fixed throughout the
subsequent flow. The sampled extent therefore determines the browsing
region and provides the conditioning context for key event generation.

\subsection{Extent-Conditioned Key Event Flow}
\label{sec:event_flow}

Given the browsing extent, BrowseFlow models the conditional
distribution of key events through a continuous-time Markov chain
(CTMC) over ordered partial event sequences, following discrete
Flow Matching for variable-length sequence
generation~\cite{gat2024discrete,havasi2026edit}.
The state at flow time $s\in[0,1]$ is a partial event sequence $A_s$,
initialized at $A_0=\varnothing$. Flow time measures generation
progress rather than physical browsing time.

For a state $A=(\tilde e_1,\ldots,\tilde e_m)$, an operation
$z=(g,d,\mu)$ inserts a key event with depth $d$ and mark $\mu$
after the first $g$ events. Let $\tilde d_0=f$ and
$\tilde d_{m+1}=D$, and let $\mathcal V_A(d)$ denote the valid
unused event marks at depth $d$. The admissible insertion space is
\begin{equation}
    \Omega(A,D)
    =\{(g,d,\mu):
    0\leq g\leq m,\
    \tilde d_g\leq d\leq\tilde d_{g+1},\
    \mu\in\mathcal V_A(d)\}.
    \label{eq:insertion_support}
\end{equation}
The neighboring depths preserve event order. In the click setting,
$\mathcal V_A(d)$ contains previously unused targets available by
depth $d$. Structural validity is therefore enforced by the
transition support.

A Transformer encodes the current partial sequence using the
extent-restricted item representations, event depths, and sequence
positions. This provides partial-sequence conditioning.
An embedding of $D$ provides extent conditioning, while a flow-time
embedding specifies the current stage of generation.
For compactness, let $\xi=(A,s,D,c_D)$ denote the shared
conditioning. The transition rate of an admissible insertion is
\begin{equation}
    r_\theta(z\mid\xi)
    =\lambda_\theta(\xi)\,\pi_\theta(z\mid\xi),
    \qquad z\in\Omega(A,D).
    \label{eq:transition_rate}
\end{equation}
The nonnegative intensity $\lambda_\theta$ determines the total
insertion rate, while the normalized distribution $\pi_\theta$
allocates that rate across valid operations. We factorize the
operation distribution as
\begin{equation}
    \pi_\theta(g,d,\mu\mid\xi)
    =\pi_\theta^g(g\mid\xi)\,
    \pi_\theta^d(d\mid g,\xi)\,
    \pi_\theta^\mu(\mu\mid g,d,\xi).
    \label{eq:operation_factorization}
\end{equation}
The three factors select an insertion gap, a browsing depth within
that gap, and an event mark. Each factor is normalized over choices
that admit a valid operation in $\Omega(A,D)$.

Conditional on the current state, an admissible operation $z$
occurs during a short interval $\Delta s$ with probability
$r_\theta(z\mid\xi)\Delta s+o(\Delta s)$. Starting from the
empty sequence, the CTMC generates a variable number of key
events without a separate event count predictor. Their locations
and marks are sampled through the insertion distribution.
At $s=1$, the resulting sequence defines a sample from
$p_\theta(\mathcal E\mid D,X_D)$. The sampled extent and
generated events are then decoded into a complete browsing
response. Each response is generated using one extent sample
and one event flow, without replanning during generation.

\subsection{Conditional Flow Matching}
\label{sec:training}

To learn the stochastic event generation process, BrowseFlow
constructs conditional probability paths from the empty sequence
to observed key event sequences. Averaging these paths over the
data distribution yields a marginal path whose dynamics can be
learned from individual logged responses.

For a training response, let $D^*$ denote the observed extent and
$\mathcal E^*=(e_1^*,\ldots,e_M^*)$ the observed key event sequence.
The event branch is conditioned on the context encoded at $D^*$.
In the following derivation, we fix $(D,X_D)$ and suppress this
conditioning, writing $p_\star(\mathcal E)$ for the conditional
data distribution.

At flow time $s\in[0,1]$, each target event is independently
retained with probability $s$, preserving its original order.
For any ordered subsequence $A\preceq\mathcal E^*$,
\begin{equation}
    p_s(A\mid\mathcal E^*)
    =s^{|A|}(1-s)^{M-|A|}.
    \label{eq:conditional_path}
\end{equation}
All other sequences have zero probability. The endpoints satisfy
$p_0(\cdot\mid\mathcal E^*)=\delta_{\varnothing}$ and
$p_1(\cdot\mid\mathcal E^*)=\delta_{\mathcal E^*}$.
The path progressively reveals the target events without imposing
a left-to-right generation order.

This path admits a continuous-time insertion interpretation.
Assign each target event an independent activation time
$U_k\sim\mathcal U(0,1)$, retaining it when $U_k\leq s$.
For an event still absent at time $s$,
\begin{equation}
    \Pr(s<U_k\leq s+\Delta\mid U_k>s)
    =\frac{\Delta}{1-s}.
    \label{eq:activation_hazard}
\end{equation}
Thus, each missing event has conditional insertion hazard
$1/(1-s)$. Let $\mathcal M(A,\mathcal E^*)$ denote the missing
target indices. Each missing event $k$ corresponds to a unique
insertion $z_k^*(A)$ into the current partial sequence, with
conditional rate
\begin{equation}
    r_s^*\!\left(z_k^*(A)\mid A,\mathcal E^*\right)
    =\frac{1}{1-s},
    \qquad k\in\mathcal M(A,\mathcal E^*).
    \label{eq:conditional_rate}
\end{equation}
All other conditional insertion rates are zero. Writing
$R=|\mathcal M(A,\mathcal E^*)|$, the total conditional intensity
is $R/(1-s)$.

Each observed sequence defines a conditional flow. Their mixture
induces the marginal probability path
\begin{equation}
    p_s(A)=\mathbb E_{\mathcal E^*\sim p_\star}
    \!\left[p_s(A\mid\mathcal E^*)\right].
    \label{eq:marginal_path}
\end{equation}
At states of positive marginal probability, the corresponding
transition rates are posterior averages of the conditional rates:
\begin{equation}
    \bar r_s(z\mid A)=
    \mathbb E\!\left[
        r_s^*(z\mid A,\mathcal E^*)\mid A,s
    \right].
    \label{eq:marginal_rate}
\end{equation}
The posterior is induced by
$p_s(A\mid\mathcal E^*)p_\star(\mathcal E^*)$.
Averaging the conditional Kolmogorov forward equations establishes
that $\bar r_s$ generates the marginal path
$p_s$~\cite{gat2024discrete,havasi2026edit}.
This connects supervision from individual observed responses
with the dynamics of the conditional event distribution.

We learn the transition field using the generalized
Kullback--Leibler divergence $\mathcal K$ between nonnegative
rate vectors, defined by
$\mathcal K(a\Vert b)=\sum_z[a_z\log(a_z/b_z)-a_z+b_z]$.
Since $\bar r_s$ is the posterior mean of the conditional rates,
\begin{equation}
    \mathbb E\!\left[
        \mathcal K(r_s^*\Vert r_\theta)\mid A,s
    \right]
    =\mathcal K(\bar r_s\Vert r_\theta)+C_s(A),
    \label{eq:rate_matching_equivalence}
\end{equation}
where $C_s(A)$ is independent of the model parameters.
Conditional rate matching thus optimizes the same
parameter-dependent objective as matching the marginal rates.

For a sampled partial state $A_s$, let
$\mathcal M=\mathcal M(A_s,\mathcal E^*)$,
$R=|\mathcal M|$, and $z_k^*=z_k^*(A_s)$.
Substituting Eq.~\eqref{eq:conditional_rate} and the factorization
$r_\theta=\lambda_\theta\pi_\theta$ into
$(1-s)\mathcal K$, and dropping model-independent terms, gives
\begin{equation}
    \mathcal L_{\mathrm{flow}}
    =(1-s)\lambda_\theta
    -R\log\lambda_\theta
    -\sum_{k\in\mathcal M}\log\pi_\theta(z_k^*).
    \label{eq:flow_loss}
\end{equation}
The first two terms fit the total insertion intensity, while
the last allocates that intensity across valid insertion operations.
The factor $(1-s)$ removes the singular target-rate scale near
$s=1$ without changing the unrestricted pointwise optimum.

Combining extent prediction with conditional rate matching yields
\begin{equation}
    \mathcal L(\Theta)=
    \mathbb E\!\left[
        -\log\pi_\eta(D^*\mid X)
        +\mathcal L_{\mathrm{flow}}
    \right].
    \label{eq:training_objective}
\end{equation}
The expectation is over logged responses, uniformly sampled
flow times, and partial event sequences drawn from
Eq.~\eqref{eq:conditional_path}. The objective jointly learns
uncertainty in browsing extent and the conditional distribution
of key events.

The marginal intensity further characterizes variable event counts.
Let $\bar\lambda_s(A)=\sum_z\bar r_s(z\mid A)$.
Summing the conditional insertion rates gives
\begin{equation}
    (1-s)\bar\lambda_s(A)=\mathbb E[R\mid A,s].
    \label{eq:intensity_interpretation}
\end{equation}
At the unrestricted population optimum, the scaled intensity
estimates the expected number of missing events. Event cardinality
is therefore learned through the insertion process rather than
predicted by a separate count head.

Algorithm~\ref{alg:browseflow_train} summarizes the training
procedure. Each observed response provides a randomly retained
partial sequence and its missing event insertion targets.
The resulting stochastic loss jointly supervises browsing extent
and conditional insertion rates.

\begin{algorithm}[t]
\caption{Conditional Rate Matching Training}
\label{alg:browseflow_train}
\begin{algorithmic}[1]
\Require Context $X$, observed response $\tau^*$
\Ensure Stochastic training loss $\ell$

\State $(D^*,\mathcal E^*)\gets\mathcal R_X(\tau^*)$
\State $c_{D^*}\gets\textsc{EventEncode}(X,D^*)$
\State $s\sim\mathcal U(0,1)$
\State $A_s\gets\textsc{Retain}(\mathcal E^*,s)$
\State $\mathcal M\gets\mathcal M(A_s,\mathcal E^*)$
\State $R\gets|\mathcal M|$
\State $\xi\gets(A_s,s,D^*,c_{D^*})$
\State $\lambda\gets\lambda_\theta(\xi)$
\State $\ell_{\mathrm{flow}}\gets
    (1-s)\lambda-R\log\lambda$

\For{$k\in\mathcal M$}
    \State $z_k^*\gets
        \textsc{TargetOp}(k,A_s,\mathcal E^*)$
    \State $\ell_{\mathrm{flow}}\gets
        \ell_{\mathrm{flow}}
        -\log\pi_\theta(z_k^*\mid\xi)$
\EndFor

\State $\ell\gets-\log\pi_\eta(D^*\mid X)
    +\ell_{\mathrm{flow}}$
\State \Return $\ell$
\end{algorithmic}
\end{algorithm}

\section{Experiments}
\subsection{Experimental Setup}
\label{sec:exp_setup}

\textbf{Datasets and tasks.}
We evaluate BrowseFlow on three datasets covering complete browsing
response generation and key event generation under fixed exposures.
\textbf{TaobaoBrowse} is constructed from real Taobao recommendation
logs and serves as our primary benchmark. It preserves complete
browsing trajectories, including progression actions, key events,
and session termination, together with user and item context.
\textbf{KuaiRand-Like}~\cite{gao2022kuairand} uses likes as key
events over fixed observed exposure sequences.
\textbf{AdSERP}~\cite{latifzadeh2025versatile} records fine-grained
interactions with search result pages and is adapted to click event
generation over fixed exposures. Table~\ref{tab:dataset_stats}
summarizes the dataset statistics.

\begin{table}[t]
    \centering
    \caption{Dataset statistics. Key rate denotes the fraction
    of valid exposures with positive key events across all splits.
    Key events correspond to clicks on TaobaoBrowse and AdSERP,
    and likes on KuaiRand-Like.}
    \label{tab:dataset_stats}
    \normalsize
    \setlength{\tabcolsep}{3pt}
    \renewcommand{\arraystretch}{1.08}
    \begin{tabular}{@{}lrrrr@{}}
        \toprule
        Dataset & Train & Val. & Test & Key rate \\
        \midrule
        TaobaoBrowse
        & 100K & 10K & 20K & 5.48\% \\
        KuaiRand-Like
        & 100K & 5K & 10K & 1.52\% \\
        AdSERP
        & 1,613 & 355 & 338 & 5.96\% \\
        \bottomrule
    \end{tabular}
\end{table}

\textbf{Baselines.}
We compare BrowseFlow with six representative approaches to user
response modeling. UBM~\cite{dupret2008user} and
DBN~\cite{chapelle2009dynamic} model examination and interaction
over ranked items. RL4RS-DIEN~\cite{wang2023rl4rs,zhou2019deep}
captures evolving user interests, while UserSim~\cite{zhao2021usersim}
and KuaiSim~\cite{zhao2023kuaisim} model sequential user behavior.
We also include CaseRec~\cite{zhao2025improving}, which uses a
Transformer to predict feedback over exposure sequences.

\textit{Energy Score (ES).}
We use the energy score~\cite{gneiting2007strictly}, a proper
scoring rule for predictive distributions, to evaluate stochastic
user responses. For context $i$, let $y_i$ denote the observed
response and $x_{i1},\ldots,x_{iK}$ denote independent generated
samples. The estimated energy score is
\begin{equation}
    \widehat{\mathrm{ES}}_i
    =\frac{1}{K}\sum_k d(x_{ik},y_i)
    -\frac{1}{K(K-1)}\sum_{k<l}d(x_{ik},x_{il}).
    \label{eq:metric_es}
\end{equation}
The reported ES averages across evaluation contexts.
The first term measures fidelity to observed responses,
while the second accounts for variation among generated responses.

For TaobaoBrowse, each response suffix is encoded as a fixed-length
sequence $\xi(\tau)$ containing action tokens, click targets,
padding, and a terminal coordinate. We use the square root
of normalized token mismatch:
\begin{equation}
    d_{\mathrm{TB}}(\tau,\tau')
    =
    \sqrt{
        \frac{1}{B}
        \sum_{r=1}^{B}
        \mathbf{1}\{\xi_r(\tau)\neq\xi_r(\tau')\}
    },
    \label{eq:distance_taobao}
\end{equation}
where $B$ is the number of encoded coordinates.
For KuaiRand-Like and AdSERP, responses are binary key event
vectors over the same valid exposure sequence. Their distance is
\begin{equation}
    d_{\mathrm{bin}}(x,y)
    =
    \sqrt{
        \frac{1}{L}
        \sum_{j=1}^{L}(x_j-y_j)^2
    },
    \label{eq:distance_binary}
\end{equation}
where $L$ denotes the number of valid exposures.
Lower ES indicates better agreement with the observed
response distribution.

\textit{Count W1.}
For one-dimensional distributions, the Wasserstein distance
$W_1$ equals the integral of the absolute difference
between their cumulative distribution functions.
Let $F_{\mathrm{obs}}$ and $F_{\mathrm{gen}}$ denote the
empirical cumulative distributions of observed and generated
key event counts. Since key event counts are nonnegative
integers, their CDFs are constant between consecutive
integers, yielding
\begin{equation}
\begin{aligned}
    \mathrm{Count\ W1}
    &=\int_0^\infty
    |F_{\mathrm{obs}}(t)-F_{\mathrm{gen}}(t)|\,\mathrm{d}t\\
    &=\sum_{k=0}^{\infty}
    |F_{\mathrm{obs}}(k)-F_{\mathrm{gen}}(k)|.
\end{aligned}
\label{eq:metric_count_w1}
\end{equation}
Count W1 measures the agreement of aggregate key event
count distributions across evaluation contexts, with
lower values indicating better alignment. Additionally, we report AUC for local response discrimination,
perplexity (PPL) for probabilistic fit, and SeqHit@1 for exact
response matching.

\textbf{Implementation.}
BrowseFlow is implemented in PyTorch 2.10.0.
All experiments were conducted using ZW810E (96\,GB) PPUs. The model uses 128-dimensional hidden states
and 2, 2, and 4 Transformer layers for user history,
session context, and partial key event encoding, respectively.
We optimize using AdamW with an effective batch size of 128,
weight decay of 0.01, and a warmup phase covering the first
5\% of training updates. Learning rates are
$3\times10^{-4}$ for TaobaoBrowse and KuaiRand-Like,
and $10^{-4}$ for AdSERP.
Models are trained on the designated training splits,
with the checkpoint minimizing validation Energy Score selected
for evaluation. Generation discretizes flow time into
32 intervals. For distributional evaluation, we sample
8 response trajectories per context on TaobaoBrowse and
128 response vectors per context on the fixed-exposure tasks.

\begin{table*}[t]
    \centering
    \caption{Main results.
    Best results are in \textbf{bold} and second-best results are
    \underline{underlined}.}
    \label{tab:main_results}
    \normalsize
    \renewcommand{\arraystretch}{1.08}
    \setlength{\tabcolsep}{7pt}
    \begin{tabular}{llccccc}
        \toprule
        Dataset & Method
        & ES $\downarrow$
        & Count W1 $\downarrow$
        & PPL $\downarrow$
        & SeqHit@1 $\uparrow$
        & AUC $\uparrow$ \\
        \midrule

        \multirow{7}{*}{TaobaoBrowse}
        & KuaiSim
        & 0.11066 $\pm$ 0.00086
        & 0.3675 $\pm$ 0.0043
        & 2.1249 $\pm$ 0.0019
        & 0.0430 $\pm$ 0.0009
        & \underline{0.6496 $\pm$ 0.0078} \\
        & RL4RS-DIEN
        & 0.11085 $\pm$ 0.00075
        & 0.3748 $\pm$ 0.0039
        & \underline{2.1099 $\pm$ 0.0015}
        & \underline{0.0452 $\pm$ 0.0013}
        & \textbf{0.6592 $\pm$ 0.0046} \\
        & UserSim
        & 0.11027 $\pm$ 0.01177
        & \underline{0.2595 $\pm$ 0.1074}
        & 2.2962 $\pm$ 0.0595
        & 0.0355 $\pm$ 0.0078
        & 0.5370 $\pm$ 0.0080 \\
        & UBM
        & 0.10864 $\pm$ 0.00088
        & 0.3283 $\pm$ 0.0065
        & 2.1931 $\pm$ 0.0023
        & 0.0437 $\pm$ 0.0014
        & 0.6262 $\pm$ 0.0042 \\
        & DBN
        & \underline{0.10847 $\pm$ 0.00088}
        & 0.3364 $\pm$ 0.0042
        & 2.1975 $\pm$ 0.0023
        & 0.0428 $\pm$ 0.0013
        & 0.6189 $\pm$ 0.0114 \\
        & CaseRec
        & 0.11172 $\pm$ 0.00635
        & 0.3888 $\pm$ 0.0881
        & 2.2068 $\pm$ 0.0260
        & 0.0440 $\pm$ 0.0071
        & 0.6188 $\pm$ 0.0088 \\
        & BrowseFlow
        & \textbf{0.07264 $\pm$ 0.00027}
        & \textbf{0.0697 $\pm$ 0.0051}
        & \textbf{2.0906 $\pm$ 0.0063}
        & \textbf{0.0544 $\pm$ 0.0013}
        & 0.6126 $\pm$ 0.0071 \\
        \midrule

        \multirow{7}{*}{KuaiRand-Like}
        & KuaiSim
        & 0.02775 $\pm$ 0.00003
        & 0.0419 $\pm$ 0.0082
        & 1.0502 $\pm$ 0.0002
        & 0.8545 $\pm$ 0.0037
        & 0.8657 $\pm$ 0.0038 \\
        & RL4RS-DIEN
        & 0.02782 $\pm$ 0.00014
        & 0.0462 $\pm$ 0.0091
        & 1.0518 $\pm$ 0.0004
        & 0.8500 $\pm$ 0.0024
        & 0.8634 $\pm$ 0.0020 \\
        & UserSim
        & \underline{0.02736 $\pm$ 0.00009}
        & 0.0515 $\pm$ 0.0138
        & \textbf{1.0491 $\pm$ 0.0002}
        & 0.8554 $\pm$ 0.0037
        & 0.8785 $\pm$ 0.0050 \\
        & UBM
        & 0.02743 $\pm$ 0.00001
        & 0.0569 $\pm$ 0.0025
        & 1.0497 $\pm$ 0.0001
        & 0.8550 $\pm$ 0.0046
        & \textbf{0.8873 $\pm$ 0.0017} \\
        & DBN
        & 0.02846 $\pm$ 0.00003
        & 0.1017 $\pm$ 0.0009
        & 1.0509 $\pm$ 0.0000
        & 0.8513 $\pm$ 0.0031
        & 0.8621 $\pm$ 0.0027 \\
        & CaseRec
        & 0.02864 $\pm$ 0.00012
        & \underline{0.0363 $\pm$ 0.0153}
        & 1.0533 $\pm$ 0.0008
        & \textbf{0.8622 $\pm$ 0.0063}
        & 0.8245 $\pm$ 0.0094 \\
        & BrowseFlow
        & \textbf{0.02703 $\pm$ 0.00007}
        & \textbf{0.0207 $\pm$ 0.0039}
        & \underline{1.0495 $\pm$ 0.0004}
        & \underline{0.8598 $\pm$ 0.0031}
        & \underline{0.8853 $\pm$ 0.0057} \\
        \midrule

        \multirow{7}{*}{AdSERP}
        & KuaiSim
        & 0.17554 $\pm$ 0.00079
        & 0.5758 $\pm$ 0.0074
        & 1.1991 $\pm$ 0.0009
        & 0.1052 $\pm$ 0.0027
        & 0.7945 $\pm$ 0.0012 \\
        & RL4RS-DIEN
        & \underline{0.15218 $\pm$ 0.00059}
        & 0.6332 $\pm$ 0.0102
        & \underline{1.1719 $\pm$ 0.0006}
        & \underline{0.1384 $\pm$ 0.0058}
        & \textbf{0.8727 $\pm$ 0.0020} \\
        & UserSim
        & 0.16912 $\pm$ 0.00103
        & \underline{0.4626 $\pm$ 0.0496}
        & 1.1906 $\pm$ 0.0023
        & 0.1314 $\pm$ 0.0038
        & 0.8058 $\pm$ 0.0034 \\
        & UBM
        & 0.17479 $\pm$ 0.00055
        & 0.8696 $\pm$ 0.0117
        & 1.2011 $\pm$ 0.0005
        & 0.0909 $\pm$ 0.0008
        & 0.8062 $\pm$ 0.0027 \\
        & DBN
        & 0.20187 $\pm$ 0.00110
        & 0.5872 $\pm$ 0.0195
        & 1.4210 $\pm$ 0.0107
        & 0.0533 $\pm$ 0.0056
        & 0.6029 $\pm$ 0.0087 \\
        & CaseRec
        & 0.17822 $\pm$ 0.00160
        & 0.5565 $\pm$ 0.0279
        & 1.2016 $\pm$ 0.0013
        & 0.1096 $\pm$ 0.0012
        & 0.7713 $\pm$ 0.0042 \\
        & BrowseFlow
        & \textbf{0.14271 $\pm$ 0.00243}
        & \textbf{0.1215 $\pm$ 0.0170}
        & \textbf{1.1344 $\pm$ 0.0013}
        & \textbf{0.2216 $\pm$ 0.0094}
        & \underline{0.8553 $\pm$ 0.0046} \\
        \bottomrule
    \end{tabular}
\end{table*}

\subsection{Results}
\label{sec:results}

\textbf{TaobaoBrowse.}
BrowseFlow achieves the strongest distributional performance,
leading in both ES and Count W1. These results indicate
better agreement with complete browsing response distributions
and closer alignment of key event count distributions.
It also obtains the best PPL and SeqHit@1, demonstrating
improved probabilistic fit and exact response matching.
Its AUC, however, remains below that of several local
response predictors. While these methods directly estimate
click feedback at individual browsing states, BrowseFlow
models the conditional distribution of complete responses
through browsing extent and key event generation.
This distinction highlights the different objectives
of local click discrimination and stochastic response
modeling.

\textbf{KuaiRand-Like.}
BrowseFlow maintains the lowest ES and Count W1 on
KuaiRand-Like, with particularly pronounced improvements
in key event count alignment. Unlike TaobaoBrowse, this
benchmark provides fixed exposure sequences and evaluates
the generation of likes as key events. The smaller
improvement in ES indicates a more modest gain in overall
response fidelity, whereas the advantage in event count
alignment remains substantial. The performance of local
prediction remains competitive across all three metrics,
demonstrating the effectiveness of the proposed formulation
in fixed exposures.

\textbf{AdSERP.}
BrowseFlow achieves the best ES and Count W1,
alongside leading performance in probabilistic fit
and exact response matching. It also maintains
competitive local discrimination. AdSERP captures interactions with search results,
providing a distinct exposure setting. Consistent
improvements in distributional metrics further support
the applicability of stochastic key event generation
across different interaction environments.

Across all three datasets, BrowseFlow consistently achieves
the best ES and Count W1.
The improvements in ES indicate better agreement with
stochastic response distributions, while the reductions
in Count W1 demonstrate closer alignment of aggregate
key event count distributions.
These results also highlight that strong local prediction
performance does not necessarily imply accurate modeling
of the conditional browsing response distribution.

\begin{figure*}[t]
    \centering
    \begin{minipage}[t]{0.33\textwidth}
        \centering
        \includegraphics[width=\linewidth]{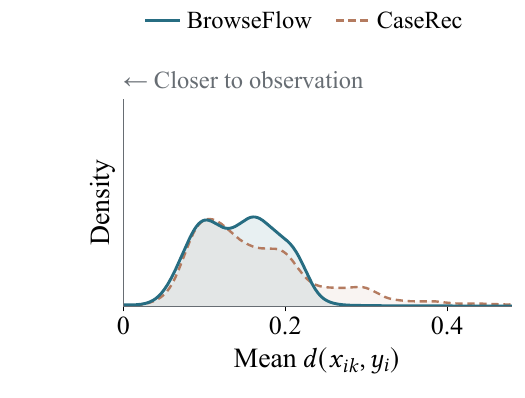}
        \par\vspace{2pt}
        {\footnotesize (a) Distance to observation\par}
    \end{minipage}\hfill
    \begin{minipage}[t]{0.33\textwidth}
        \centering
        \includegraphics[width=\linewidth]{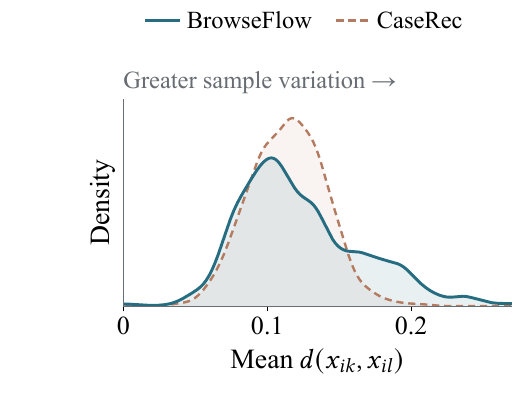}
        \par\vspace{2pt}
        {\footnotesize (b) Distance between samples\par}
    \end{minipage}\hfill
    \begin{minipage}[t]{0.33\textwidth}
        \centering
        \includegraphics[width=\linewidth]{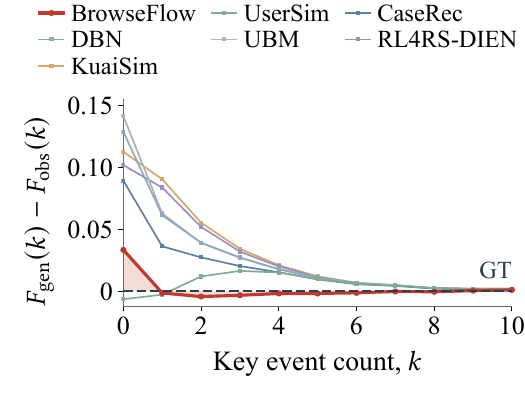}
        \par\vspace{2pt}
        {\footnotesize (c) Key event count CDF deviation\par}
    \end{minipage}
    \caption{\textbf{Distributional analysis on TaobaoBrowse.}
    (a) Distribution of mean distances between generated
    and observed responses, $d(x_{ik},y_i)$, across contexts.
    (b) Distribution of mean pairwise distances between
    distinct generated responses, $d(x_{ik},x_{il})$,
    across contexts. 
    (c) Signed differences between generated and observed
    key event count cumulative distribution functions (CDFs),
    $F_{\mathrm{gen}}(k)-F_{\mathrm{obs}}(k)$,
    for all seven methods.}

    \label{fig:distributional_analysis}
\end{figure*}

\subsection{Ablation Study}
\label{sec:ablation}

We investigate the effects of response representation and
generation mechanism on BrowseFlow. Table~\ref{tab:ablation}
compares dense action sequences and sparse key event
representations under autoregressive generation and two
Flow Matching variants.

\begin{figure*}[t]
    \centering
    \begin{minipage}[t]{0.45\textwidth}
        \centering
        \includegraphics[width=\linewidth]{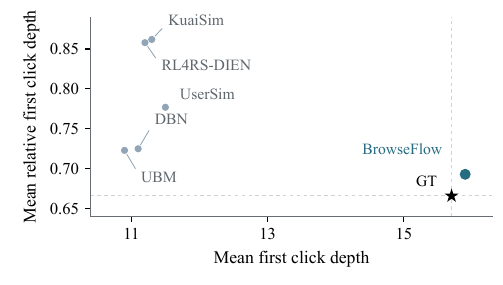}
        \par\vspace{2pt}
        {\footnotesize (a) First click placement\par}
    \end{minipage}\hfill
    \begin{minipage}[t]{0.45\textwidth}
        \centering
        \includegraphics[width=\linewidth]{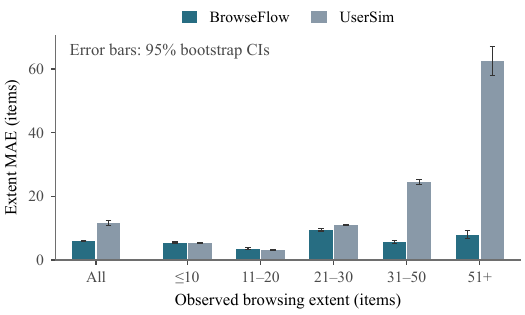}
        \par\vspace{2pt}
        {\footnotesize (b) Browsing extent recovery\par}
    \end{minipage}
        
    \caption{\textbf{Behavioral structure analysis on TaobaoBrowse.}
    (a) Mean first click placement in absolute and normalized
    browsing depth, compared with observed responses.
    (b) Browsing extent MAE grouped by observed depth,
    comparing BrowseFlow with UserSim. Lower MAE indicates
    more accurate extent recovery.}
    \label{fig:behavioral_structure}
\end{figure*}

Under autoregressive generation, sparse representation
improves both ES and Count W1. By summarizing
frequent progression actions through browsing extent,
the sparse representation focuses generation on key events.
These improvements suggest that separating browsing
progression from key event generation benefits both
response fidelity and event count alignment.

Flow Matching further improves both ES and Count W1
under dense and sparse representations, with particularly
large gains in key event count alignment.
Through conditional rate matching, it learns stochastic
transition dynamics over partial response sequences,
directly modeling event generation within the response
structure. These improvements highlight the importance
of the generation mechanism in capturing stochastic
response distributions, beyond the choice of
representation alone.

The effect of insertion depends on the response
representation. Append-only generation follows browsing
order, whereas insertion allows events to be placed between
previously generated events. In the sparse key event space,
insertion further improves Count W1, with a relatively small
change in ES. In contrast, dense insertion degrades both
metrics relative to append-only generation, with a particularly
large deterioration in key event count alignment.
These results highlight the complementary roles
of sparse representation and insertion-based Flow Matching
in generating structured browsing responses under
explicit exposure constraints.

\begin{table}[t]
    \centering

    \caption{Component ablation on TaobaoBrowse.
    Sparse denotes the browsing extent and key event
    representation. FM denotes Flow Matching, with
    append-only or insertion generation.
    The final row corresponds to BrowseFlow.}

    \label{tab:ablation}
    \normalsize
    \setlength{\tabcolsep}{5pt}
    \renewcommand{\arraystretch}{1.12}
    \begin{tabular}{@{}ccc@{\hspace{7pt}\vrule width 0.5pt\hspace{7pt}}cc@{}}
        \toprule
        Sparse & FM & Insertion
        & ES $\downarrow$ & Count W1 $\downarrow$ \\
        \midrule

        \abNo & \abNo & \abNo
        & 0.07815 & 0.45158 \\

        \abYes & \abNo & \abNo
        & 0.07444 & 0.38071 \\

        \addlinespace[2pt]

        \abNo & \abYes & \abNo
        & 0.07392 & 0.12666 \\

        \abYes & \abYes & \abNo
        & \underline{0.07279} & \underline{0.09407} \\

        \abNo & \abYes & \abYes
        & 0.07786 & 2.56907 \\

        \midrule

        \abYes & \abYes & \abYes
        & \textbf{0.07264} & \textbf{0.06969} \\

        \bottomrule
    \end{tabular}
\end{table}

\subsection{Distributional Analysis}
\label{sec:distributional_analysis}

\textbf{Response fidelity and generated variability.}
Figure~\ref{fig:distributional_analysis}(a) examines response
fidelity through the mean distance between generated and
observed responses. Compared with CaseRec, BrowseFlow
produces smaller sample-to-observation distances and a
shorter upper tail, indicating fewer contexts with large
deviations from observed behavior.
Figure~\ref{fig:distributional_analysis}(b) examines generated
variability through the mean pairwise distance between
responses generated under the same context. BrowseFlow's
sample-to-sample distances remain nonzero and exhibit a
more pronounced upper tail, showing that repeated
generations retain behavioral variation.
Energy Score evaluates these two aspects jointly,
reflecting the balance between observation fidelity
and generated variability rather than favoring
either in isolation.

\textbf{Key event count alignment.}
Count W1 measures the agreement between observed and generated
key event count distributions. Figure~\ref{fig:distributional_analysis}(c)
provides a more detailed comparison through signed CDF
deviations, where the zero line indicates exact agreement
in cumulative probability at each event count. Positive
deviations indicate excessive probability assigned to
responses with at most that many key events.
In particular, the deviation at zero directly reflects
the mismatch in the probability of browsing responses
without any key events.
BrowseFlow remains closer to zero across much of the
count range, although some baselines achieve smaller
deviations at individual counts. 
This pattern is consistent with its leading Count W1
and indicates more accurate recovery of the observed
key event count distribution.


\subsection{Behavioral Structure Analysis}
\label{sec:behavioral_analysis}

\textbf{First click placement.}
Figure~\ref{fig:behavioral_structure}(a) compares the mean
first click positions in absolute and normalized browsing
depth. Absolute depth reflects where the first click occurs
along the exposure sequence, while normalized position
accounts for differences in browsing extent. BrowseFlow
more closely matches the observed placement on both
measures, whereas competing methods generally place
the first click at shallower absolute depths.
The agreement in normalized position further suggests
that this improvement is not explained solely by
differences in generated browsing extent.

\textbf{Browsing extent recovery.}
Figure~\ref{fig:behavioral_structure}(b) compares browsing
extent errors across ranges defined by observed depth.
BrowseFlow and UserSim perform similarly on shorter
sessions, but their errors diverge as browsing depth
increases. BrowseFlow maintains substantially lower
errors on longer sessions, with the largest differences
appearing in the deepest ranges. UserSim tends to
underestimate the extent of these sessions, whereas
BrowseFlow more closely recovers their browsing lengths.


\section{Conclusion}
In this work, we studied generative user simulation under two characteristics of browsing behavior:
\textbf{Sparse Key Events} and \textbf{Stochastic Responses}.
To address these challenges, we proposed BrowseFlow, which
models the conditional distribution of browsing responses
through probabilistic extent prediction and discrete flow
matching over key events. This formulation separates
browsing progression from sparse key event generation
and models event counts and locations under exposure constraints.
We also introduced TaobaoBrowse, a benchmark pairing real
recommendation contexts with fine-grained browsing trajectories.

Experiments on TaobaoBrowse and two external benchmarks
demonstrate consistent improvements in Energy Score and
Count W1, reflecting better agreement with 
response distributions.
Further analyses reveal improved observation fidelity
while preserving variation across generated responses,
alongside more accurate browsing extent recovery and
first click placement. 
These findings highlight a central principle of faithful
user simulation: model the key events that matter,
and capture the many ways users may browse.

% ...
\endgroup
\bibliographystyle{ACM-Reference-Format}
\bibliography{references}


\providecommand{\benchYes}{\textcolor{green!48!black}{\ensuremath{\checkmark}}}
\providecommand{\benchNo}{\textcolor{red!70!black}{\ensuremath{\times}}}

\appendix


\section{Benchmark and Dataset Details}
\label{app:data_processing}

\subsection{TaobaoBrowse Construction}

Browsing extent follows the order of factual exposures and is
distinguished from the original candidate position. Unexposed
candidates are not labeled as negative feedback.

After the observed prefix, 56.76\% of response suffixes contain
no key event and 25.70\% contain exactly one. On TaobaoBrowse,
key events correspond to clicks. Figure~\ref{fig:taobao_statistics}
shows the distributions of browsing extent and key event count,
together with key event rates and exposure support across
browsing depths.

\begin{figure*}[!t]
    \centering
    \begin{subfigure}[t]{0.33\textwidth}
        \centering
        \includegraphics[width=\linewidth]{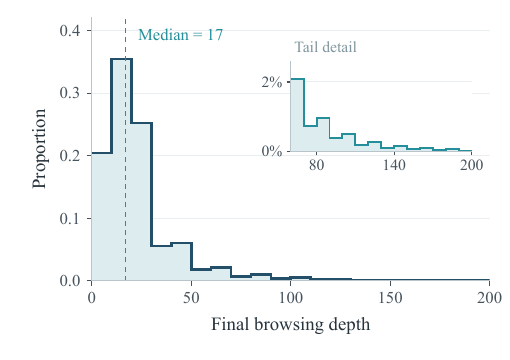}
        \caption{Browsing extent.}
        \label{fig:taobao_extent}
    \end{subfigure}
    \hfill
    \begin{subfigure}[t]{0.33\textwidth}
        \centering
        \includegraphics[width=\linewidth]{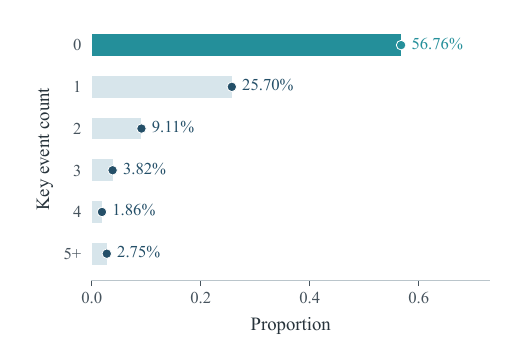}
        \caption{Key event count.}
        \label{fig:taobao_event_count}
    \end{subfigure}
    \begin{subfigure}[t]{0.33\textwidth}
        \centering
        \includegraphics[width=\linewidth]{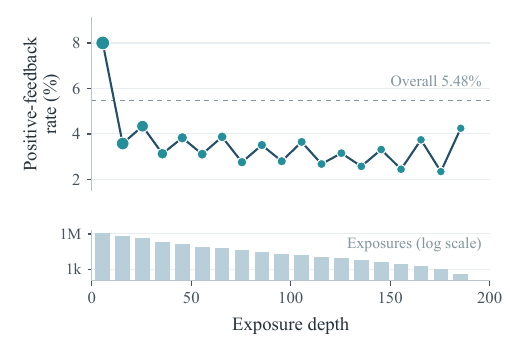}
        \caption{Key event rate and exposure support.}
        \label{fig:taobao_feedback_depth}
    \end{subfigure}
    \caption{Empirical browsing statistics of TaobaoBrowse.
    Key event rate denotes the fraction of valid exposures
    with a click.}
    \Description{The panels show the empirical distributions
    of browsing extent and post-prefix key event counts,
    followed by key event rates and exposure counts
    across browsing depths.}
    \label{fig:taobao_statistics}
\end{figure*}

\subsection{Comparison with Existing Benchmarks}

Table~\ref{tab:benchmark_comparison} compares five documented
properties relevant to browsing response simulation: production
provenance, factual within-session exposure progression,
clicked-target occurrence, and user-profile and item context.
Exposure progression refers to recorded browsing or viewport
advancement, rather than merely time-ordered interactions.
TaobaoBrowse combines production recommendation context with
observed progression and item-directed key events. A cross
indicates that a property is not established at the required
granularity, rather than necessarily absent from the original
dataset.

\begin{table}[H]
\centering
\caption{Documented benchmark capabilities for browsing
response simulation. \benchYes: verified;
\benchNo: not established at the required granularity.}
\label{tab:benchmark_comparison}
\normalsize
\setlength{\tabcolsep}{3pt}
\renewcommand{\arraystretch}{1.08}
\begin{tabular}{@{}lccccc@{}}
\toprule
Benchmark & Prod. & \shortstack{Exposure\\trace}
& \shortstack{Click\\target} & Profile & Item \\
\midrule
KuaiRand~\cite{gao2022kuairand}
& \benchYes & \benchNo & \benchYes & \benchYes & \benchYes \\
RecFlow~\cite{liu2025recflow}
& \benchYes & \benchNo & \benchYes & \benchYes & \benchYes \\
RL4RS~\cite{wang2023rl4rs}
& \benchYes & \benchNo & \benchYes & \benchYes & \benchYes \\
OPeRA~\cite{wang2026opera}
& \benchNo & \benchYes & \benchYes & \benchYes & \benchYes \\
AdSERP~\cite{latifzadeh2025versatile}
& \benchNo & \benchYes & \benchYes & \benchNo & \benchYes \\
\midrule
\textbf{TaobaoBrowse}
& \benchYes & \benchYes & \benchYes & \benchYes & \benchYes \\
\bottomrule
\end{tabular}
\end{table}

\subsection{External Evaluation Tasks}

\textbf{KuaiRand-Like.}
Blocks are constructed from KuaiRand-1K~\cite{gao2022kuairand}.
Records are ordered by user, date, time, and stable tie-breaking
fields. Runs are separated by date/tab changes or gaps exceeding
30 minutes, then partitioned into blocks of at most 32 exposures.
The target is \texttt{is\_like}, with likes treated as key events.
Valid positions define the binary event response vector.
The splits contain 100,000/5,000/10,000 blocks.

\textbf{AdSERP.}
Of 2,776 recorded search trials from 47 participants
\cite{latifzadeh2025versatile}, 2,306 retain reconstructible
exposure geometry and click-to-result mapping. A result is
exposed when at least half of its card is visible; re-entry
defines a new exposure occurrence. Clicks are treated as
key events. The participant-disjoint training/validation/test
split is 1,613/355/338. Both external tasks generate key
event responses conditioned on fixed observed exposure
sequences, rather than predicting natural session termination.



\section{Evaluation Details}
\label{app:metrics}

\subsection{Distributional Metrics}

TaobaoBrowse uses all 20,000 test contexts and $K=8$ generated
responses per context; the fixed-exposure tasks use $K=128$.
Energy Score (ES) is computed per context using the
U-statistic sample-pair term and then averaged across contexts.

For TaobaoBrowse, the response encoding $\xi(\tau)$ in
Eq.~\eqref{eq:distance_taobao} contains 400 ordered positions
with joint action and clicked-target categories, padded as
needed, and one terminal coordinate. A click and its target
form a joint category, while progression actions do not
carry clicked targets. Thus, $B=401$ in the response distance.

For the fixed-exposure tasks, the binary distance in
Eq.~\eqref{eq:distance_binary} is computed only over valid
exposures. Equivalently, a valid-exposure mask $m_j$ excludes
padded positions, with $L=\sum_j m_j$.

The canonical TaobaoBrowse response encoding admits an
equivalent Euclidean representation. For concatenated one-hot
vectors $\psi(\tau)$,
\begin{equation}
 d_{\mathrm{TB}}(\tau,\tau')=
 \frac{\|\psi(\tau)-\psi(\tau')\|_2}{\sqrt{802}}.
 \label{eq:app_es_embedding}
\end{equation}
The encoding is injective on valid recorded responses.
The binary response distance $d_{\mathrm{bin}}$ is likewise
a positive rescaling of Euclidean distance over the valid
exposure positions. Both distances therefore induce
strictly proper population Energy Scores on their respective
response spaces, as shown below.

\paragraph{Strict propriety of the Energy Score.}
Let $P$ be the predictive distribution and $Q$ the true
conditional response distribution on either response space.
For a distance $d$ induced by a Euclidean embedding, the
population Energy Score is
\begin{equation}
\mathrm{ES}(P,Q)
=
\mathbb{E}_{X\sim P,\,Y\sim Q}[d(X,Y)]
-\frac{1}{2}\mathbb{E}_{X,X'\sim P}[d(X,X')],
\label{eq:app_population_es}
\end{equation}
where $X$ and $X'$ are independent draws from $P$ and
$Y$ is an independent draw from $Q$. Under the stated
distance conditions, the Energy Score is strictly proper:
its expected value is uniquely minimized when $P=Q$
\cite{gneiting2007strictly}.

For TaobaoBrowse, injectivity of $\psi$ ensures that distinct
valid responses have distinct embeddings. For the fixed-exposure
tasks, each binary response is represented by its valid
exposure coordinates, so distinct response vectors also have
distinct embeddings. In both cases, the distance is a positive
multiple of Euclidean distance. Hence the population Energy
Score distinguishes different predictive distributions and
is uniquely minimized by the true conditional response
distribution.

Count W1 compares the aggregate empirical distributions
of observed and generated key event counts, rather than
averaging context-wise Wasserstein distances. Counts refer
to suffix clicks on TaobaoBrowse, likes on KuaiRand-Like,
and clicks on AdSERP, restricted to valid exposures.
For evaluations with multiple generation seeds, Count W1
is computed per seed before averaging.

\subsection{Likelihood and Local Metrics}

On TaobaoBrowse, PPL evaluates the observed next joint
response, including the clicked target when applicable,
rather than the likelihood of the complete response suffix:
\begin{equation}
 \mathrm{PPL}=
 \exp\!\left[-\frac1Q\sum_{i=1}^{Q}
 \log p(r_i\mid c_i)\right],
 \quad Q=20{,}000.
 \label{eq:app_taobao_ppl}
\end{equation}

BrowseFlow obtains local probabilities by marginalizing
over possible browsing extents and future insertion paths,
using $G=32$ flow-time intervals and $K=64$ importance
samples for nontrivial marginals. External PPL uses joint
binary response likelihood normalized by the total valid
exposures in the evaluated panel.

On TaobaoBrowse, SeqHit@1 requires exact matching of
the first sampled ordered response suffix, including
clicked targets and the terminal state. For the external
tasks, marginal key event probabilities estimated from
$K=128$ draws are thresholded at $0.5$ before exact
vector matching.

TaobaoBrowse AUC pools 1,024 prefix queries and uses the
fraction of $K=64$ first-event samples that are clicks
as the prediction score. On the fixed-exposure tasks,
AUC evaluates local key event discrimination over valid
exposures.

\subsection{Distributional and Behavioral Diagnostics}

For the ES component plots, the mean generated-to-observed
distance and the mean distance between distinct generated
samples are computed for each evaluated context and
averaged across available seeds before plotting their
context-level distributions. The pairwise distance is
shown without the factor of one half used in ES.

For the key event count CDF plot, each model's event
count probabilities are averaged across seeds before
cumulative summation and subtraction of the common
observed CDF. The displayed range covers counts 0--10,
while the full count support is retained. The signed
CDF deviation is a diagnostic; its absolute area is
not identified with the mean of seed-level Count W1 values.

For behavioral analysis, browsing extent MAE is grouped
by observed extent. UserSim is selected as the strongest
baseline by overall extent MAE, rather than separately
within each group. First click placement is evaluated
in absolute depth and relative to browsing extent,
conditional on responses containing a first click.


\section{Baseline Adaptations}
\label{app:baselines}

\subsection{Common Response Interface}

On TaobaoBrowse, baselines sequentially generate progression,
click, and termination actions. For clicks, they also select
a legal, visible, previously unclicked target. For baselines
with learned action and target prediction heads, the task
loss combines next-action and conditional clicked-target
cross-entropy:
\begin{equation}
\begin{aligned}
\mathcal L_{\mathrm{task}}
&=-\frac1T\sum_{t=1}^{T}
 \log p_\phi(a_t^*\mid h_t,X_t)\\
&\quad-\frac1{\max(1,C)}
 \sum_{t:a_t^*=\mathrm{CLICK}}
 \log p_\phi(v_t^*\mid a_t^*,h_t,X_t).
\end{aligned}
\label{eq:app_baseline_task_loss}
\end{equation}

Here $T$ and $C$ count supervised actions and clicks,
respectively. Native response representations and
legal-item scores condition these predictions; CaseRec
uses its own causal prediction head. For the fixed-exposure
tasks, methods generate binary key event feedback over
the recorded exposure sequence.

\subsection{Method-Specific Adaptations}

\noindent\textbf{UBM~\cite{dupret2008user}.}
The native click score factorizes as
$q_v=\alpha_v\gamma_{v,l}$, with item attractiveness
$\alpha_v$ and examination $\gamma_{v,l}$ conditioned
on rank and the previous clicked position $l$.
Generated clicks update $l$.

\par\noindent\textbf{DBN~\cite{chapelle2009dynamic}.}
Attractiveness, post-click satisfaction, and latent
examination continuation are propagated through an
active/stopped-state recursion. The latent stopped
state is not identified with observed session termination;
fixed-exposure likelihoods marginalize it by
forward recursion.

\par\noindent\textbf{RL4RS-DIEN~
\cite{wang2023rl4rs,zhou2019deep}.}
GRU history encoding, candidate attention, and AUGRU
interest evolution supply response features for the
task heads. On fixed exposure sequences, feedback is
sampled from history-conditioned Bernoulli probabilities.

\par\noindent\textbf{UserSim~\cite{zhao2021usersim}.}
The recurrent generator/discriminator retains its
adversarial, response-supervision, and reconstruction
objectives alongside $\mathcal L_{\mathrm{task}}$.
Generator and discriminator are pretrained before
joint updates; sampled feedback updates the recurrent
state.

\par\noindent\textbf{KuaiSim~\cite{zhao2023kuaisim}.}
Its history-conditioned response state drives task
predictions. On fixed exposure plans, feedback is
sampled within recorded groups and used to update
a recurrent state between groups.

\par\noindent\textbf{CaseRec~\cite{zhao2025improving}.}
A causal Transformer processes alternating state and
action-context tokens, using a recurrent summary of
completed events. The adapter predicts actions and legal
targets from this state, using up to ten completed events
per state.


\section{BrowseFlow Algorithm and Implementation}
\label{app:implementation}

\subsection{BrowseFlow Inference}
\label{app:algorithms}

On TaobaoBrowse, BrowseFlow first samples browsing extent
from the complete recommendation context and re-encodes
the session within the selected region. An insertion flow
then generates an ordered sequence of key events, which
is decoded into a complete browsing response.
Algorithm~\ref{alg:bf_sampling} summarizes this procedure.
For the fixed-exposure tasks, the exposure horizon is
given, and only key event responses are generated.

\begin{algorithm}[H]
\caption{BrowseFlow Inference}
\label{alg:bf_sampling}
\begin{algorithmic}[1]
\Require Context $X$, flow intervals $G$
\Ensure Generated response $\hat{\tau}$

\State $D\sim\pi_\eta(\cdot\mid X)$
\State $c_D\gets\textsc{EventEncode}(X,D)$
\State $A\gets\varnothing$
\For{$j=0,\ldots,G-1$}
    \State $s\gets(j+\frac12)/G,\quad\delta\gets1/G$
    \While{$\delta>0$}
        \State $\lambda\gets\lambda_\theta(A,s,D,c_D)$
        \If{$\Omega(A,D)=\varnothing$ \textbf{or} $\lambda=0$}
            \State \textbf{break}
        \EndIf
        \State $w\sim\operatorname{Exp}(\lambda)$
        \If{$w\geq\delta$}
            \State \textbf{break}
        \EndIf
        \State $z\sim\pi_\theta(\cdot\mid A,s,D,c_D)$
        \State $A\gets A\oplus z,\quad\delta\gets\delta-w$
    \EndWhile
\EndFor
\State \Return $\operatorname{Decode}_X(D,A)$
\end{algorithmic}
\end{algorithm}

The sampled extent and its encoded context remain fixed
throughout event generation. Within each flow interval,
exponential waiting times determine successive insertions.
After each insertion, transition rates and legal operations
are updated. The operation distribution follows the gap,
depth, and mark factorization in Section~\ref{sec:event_flow}.
The flow permits an empty key event sequence. Finally,
decoding reconstructs progression actions and session
termination on TaobaoBrowse, while the external tasks
retain the observed exposure horizon.

\subsection{Conditional Rate Matching}
\label{app:theory}

Fix the context and extent, suppressing them in the notation.
For the random-retention path, let
$r_s^*(z\mid A,\mathcal E^*)$ denote the conditional target
insertion rate and
$\bar r_s(z\mid A)=
\mathbb E[r_s^*(z\mid A,\mathcal E^*)\mid A,s]$
its marginal rate.
For nonnegative rate vectors on the same finite support,
write
\[
D_{\mathrm{GKL}}(a,b)
=\sum_z\left[
a_z\log\frac{a_z}{b_z}-a_z+b_z
\right].
\]

\begin{proposition}[Conditional rate matching]
\label{prop:conditional_matching}
For a state of positive marginal probability and
finite divergences,
\begin{equation}
\mathbb E[D_{\mathrm{GKL}}(r_s^*,r_\theta)\mid A,s]
=D_{\mathrm{GKL}}(\bar r_s,r_\theta)+C_s(A),
\label{eq:app_gkl_equivalence}
\end{equation}
where $C_s(A)$ is independent of model parameters.
\end{proposition}

\begin{proof}
Expanding the generalized KL divergence, the terms
depending on $r_\theta$ are
$-\sum_z\mathbb E[r_s^*(z)]\log r_\theta(z)
+\sum_zr_\theta(z)$,
identical to those of the marginal divergence.
The remaining difference is
\[
\sum_z\left\{
\mathbb E[r_s^*(z)\log r_s^*(z)]
-\bar r_s(z)\log\bar r_s(z)
\right\},
\]
which is independent of $\theta$.
The common positive weighting by $1-s$ preserves
this equivalence.
\end{proof}

\begin{corollary}[Expected remaining events]
\label{cor:remaining_events}
If $R$ is the number of target key events absent from $A$,
the marginal total intensity satisfies
\begin{equation}
(1-s)\bar\lambda_s(A)=\mathbb E[R\mid A,s].
\label{eq:app_remaining_events}
\end{equation}
\end{corollary}

Each missing target key event has conditional insertion
rate $1/(1-s)$. Summing these rates and taking the
posterior expectation gives the corollary.

\subsection{Model Configuration and Ablation Protocol}

BrowseFlow uses 128-dimensional hidden states, 2/2/4
Transformer layers for history, session, and partial
key events, four attention heads, FFN dimension 512,
and dropout $0.1$. We use AdamW with weight decay
$0.01$, effective batch size 128, and 5\% warmup.
Learning rates are $3\times10^{-4}$ on TaobaoBrowse
and KuaiRand-Like, and $10^{-4}$ on AdSERP.
Generation uses 32 flow-time intervals.

The ablation study compares dense action sequences
with a representation based on browsing extent and
sparse key events under autoregressive and
Flow Matching generation.
Within Flow Matching, append-only generation restricts
new events to the end of the partial sequence, whereas
insertion allows events to be placed between existing
events. These configurations isolate the effects of
response representation, generation mechanism, and
insertion operations on ES and Count W1.

\end{document}